\documentclass[11pt,letterpaper]{article}
\usepackage[margin=1in]{geometry}
\usepackage{authblk}
\usepackage[T1]{fontenc}
\usepackage{xspace}
\usepackage{amsmath}
\usepackage{amsthm}
\usepackage{amsfonts}
\usepackage{amssymb}
\usepackage{graphicx}
\usepackage[nothing]{algorithm}
\usepackage{algorithmicx}
\usepackage[noend]{algpseudocode}
\usepackage{array}

\def\FN{\textsf{FN}\xspace}
\def\MEB{\textsf{MEB}\xspace}

\def\FP{\textsf{FP}\xspace}

\def\APPROXMEB{\textsf{APPROX-MEB}\xspace}
\def\coreset{\textsf{coreset}\xspace}

\newcounter{thm}
\theoremstyle{plain}
\newtheorem{theorem}[thm]{Theorem}
\newtheorem{corollary}[thm]{Corollary}
\newtheorem{lemma}[thm]{Lemma}

\theoremstyle{remark}

\newcommand{\cB}{\mathcal{B}}

\title{Streaming Diameter of High-Dimensional Points}

\author[1]{ Magn\'us M. Halld\'orsson\thanks{Partially supported by Icelandic Research Fund grant 2511609.}}
\author[2]{Nicolaos Matsakis\thanks{Work done in part while the author was at Charles University. Partially supported by Icelandic Research Fund grant 2511609 and by Czech Science Foundation project 22-22997S.}}
\author[3]{Pavel Vesel\'y\thanks{Partially supported by Czech Science Foundation project 22-22997S, 
by ERC-CZ project LL2406 of the Ministry of Education of Czech Republic,
and
by Center for Foundations of Modern Computer Science (Charles Univ.\ project UNCE 24/SCI/008).}}

\affil[1]{Reykjavik University, \texttt{mmh@ru.is}}
\affil[2]{Reykjavik University, \texttt{nickmatsakis@gmail.com}}
\affil[3]{Charles University, \texttt{vesely@iuuk.mff.cuni.cz}}

\date{}

\begin{document}

\maketitle
\begin{abstract} 
We improve the space bound for streaming approximation of Diameter but also of Farthest Neighbor queries, Minimum Enclosing Ball and its Coreset, in high-dimensional Euclidean spaces. In particular, our deterministic streaming algorithms store $\mathcal{O}(\varepsilon^{-2}\log(\frac{1}{\varepsilon}))$ points. This improves by a factor of $\varepsilon^{-1}$ the previous space bound of Agarwal and Sharathkumar (SODA 2010), while offering a simpler and more complete argument. We also show that storing $\Omega(\varepsilon^{-1})$ points is necessary for a $(\sqrt{2}+\varepsilon)$-approximation of Farthest Pair or Farthest Neighbor queries.%
\end{abstract}

\section{Introduction}
In the streaming model the input data is assumed to be vast making it impossible to entirely store it. Therefore, we may store only a sketch of the input while making a few passes over it. The research on streaming algorithms has been remarkably fruitful and we now have optimal or near-optimal algorithms for counting distinct elements, frequency moments, quantiles and a plethora of other problems~(see~\cite{cormodebook} for a comprehensive exposition).

We focus on high-dimensional geometric streams, where the input $S$ consists of points in $\mathbb{R}^{d}$ for $d$ large, a topic of recent interest, e.g.,~\cite{Woodruff022, EsfandiariKMWZ24, makarychev, CzumajJK0Y22,  MahabadiRWZ20, ChenJLW22, DBLP:conf/stoc/ChenCJLW23, ChenJK23,JiangKS24}. Extent measures, such as the Diameter or the Minimum Enclosing Ball, are fundamental statistics of a set of points, having a body of work in both streaming and non-streaming settings~\cite{AgarwalHV04, Indyk, GoelIV01, AgarwalMS91, BadoiuHI02, BadoiuClarkson}.

In an influential work, Agarwal and  Sharathkumar~\cite{AS15j} gave a streaming algorithmic framework for several high-dimensional extent problems.
Their Blurred Ball Cover data structure maintains a collection of $\mathcal{O}(\frac{1}{\varepsilon^{2}}\log(\frac{1}{\varepsilon}))$ balls, whose union approximately covers the input $S$.
It was then used to approximate a number of high-dimensional extent problems.
The claim is that each ball is represented by a coreset of 
$\mathcal{O}(\frac{1}{\varepsilon})$ points of $S$, for a total space of $\mathcal{O}(\frac{1}{\varepsilon^{3}}\log(\frac{1}{\varepsilon}))$ points.
It appears though that a somewhat higher space bound is needed for the claimed approximations; see Appendix~\ref{app:AS15issue}.

We give a modified data structure, Guarded Ball Cover, building extensively on \cite{AS15j}.
It allows for both a simpler and more complete treatment, and also results in the smaller space bound of $\mathcal{O}(\frac{1}{\varepsilon^{2}}\log(\frac{1}{\varepsilon}))$ points. The improved space bound extends to all four applications: approximate Farthest Neighbor queries and maintaining approximate Diameter, Minimum Enclosing Ball and Coreset for Minimum Enclosing Ball.
This is feasible by storing only a single point per ball, along with its center and radius. Correctness arguments are simplified by also storing the first point of $S$ as a proxy for all points deleted later from memory.

We also show that $\Omega(\varepsilon^{-1})$ points need to be stored for a comparable $(\sqrt{2}+\varepsilon)$-approximation of Farthest Pair or Farthest Neighbor queries. This applies to a computational model where the algorithm must return an input point upon a query and the space is determined by the number of points stored; crucially, once a point is deleted, it cannot be retrieved.

\section{Preliminaries}
Let $S$ be a multiset of points in $\mathbb{R}^{d}$ that arrive sequentially in a stream. Upon arrival, each point $p\in S$ is either stored in memory (and, possibly, deleted later) or irrevocably discarded. We assume one-pass streaming algorithms in the insertion-only setting. By $\varepsilon\in(0,1]$ we denote an error parameter and by $\alpha>1$ an approximation guarantee.

An extent measure of a set of points computes certain statistics of either this set or a geometric shape enclosing it~\cite{AgarwalHV04}. Let $\Vert pq\rVert$ denote the Euclidean distance between points $p\in\mathbb{R}^{d}$ and $q\in\mathbb{R}^{d}$. The \emph{Diameter} is the maximum Euclidean distance between any pair of points in $S$ and the \emph{Farthest Pair} $\FP(S)$ is a pair of points of $S$ having Euclidean distance equal to the Diameter. The \emph{Farthest Neighbor} of a point $q$ is the point $p$ of greatest Euclidean distance from $q$. An $\alpha$-farthest-neighbor $\alpha$-$\FN(q)$ of a query point $q\in\mathbb{R}^{d}$ is a point $x\in S$ such that for every $p\in S$ it is $\Vert qp\rVert\leq \alpha \cdot \Vert xq\rVert$. 

By $B(c(B),r(B))$ we denote a ball centered at point $c(B)$ with radius $r(B)$. The $(1+\varepsilon)$-expansion of $B(c,r)$ is defined as $B(c, (1+\varepsilon)r)$. The \emph{Minimum Enclosing Ball} $\MEB(S)$ is the ball of minimum radius containing all points of $S$. A Ball $B$ is $\alpha$-$\MEB(S)$ if $r(B)\leq \alpha r(\MEB(S))$ and each point of $S$ is within Euclidean distance $r(B)$ from $c(B)$.

For a set of points $S$, a \emph{coreset} is a set $S'\subseteq S$ that preserves a geometric property of $S$~\cite{peledbook}. A set $S'\subseteq S$ is $\alpha$-$\coreset(S)$ for the \MEB, if each point of $S$ is contained in the $\alpha$-expansion of $\MEB(S')$. 

We compute $\alpha$-$\FN(q)$ for any query $q\in\mathbb{R}^{d}$ and maintain $\alpha$-$\FP(S)$,  $\alpha$-$\MEB(S)$ and $\alpha$-$\coreset(S)$. Agarwal and Sharathkumar~\cite{AS15j} gave streaming algorithms computing $(\sqrt{2}+\varepsilon)$-$\FN(q)$ and maintaining $(\frac{1+\sqrt{3}}{2}+\varepsilon)$-$\MEB(S)$, $(\sqrt{2}+\varepsilon)$-$\coreset(S)$ and $(\sqrt{2}+\varepsilon)$-$\FP(S)$. For these approximations, they designed the aforementioned Blurred Ball Cover data structure.

\subsection{Related Work}

The algorithm of Gonzalez~\cite{Gonzalez85} computes a 2-\MEB by storing the first point of $S$ and its farthest neighbor, located on the boundary of a ball defined by the first point as center. Zarrabi-Zadeh and Chan~\cite{ZZ} improved the guarantee of $\alpha$-\MEB to $\alpha=1.5$ by giving an one-pass algorithm that stores one ball. They also gave a lower bound of $\frac{\sqrt{2}+1}{2}\approx 1.207$ for the guarantee of any deterministic algorithm for the $\alpha$-\MEB that stores only one ball.

Bad\u{o}iu et al.~\cite{BadoiuHI02} showed that the number of coreset points approximating $(1+\varepsilon)$-$\MEB(S')$ for a set $S'$ in $\mathbb{R}^{d}$ does not depend on $d$. 
Improved algorithms were given in \cite{KumarMY03},~\cite{DBLP:conf/soda/BadoiuC03} and~\cite{BadoiuClarkson}.

Following the conference result of~\cite{AS15} that maintains a $(\frac{1+\sqrt{3}}{2}+\varepsilon)$-$\MEB(S)$, Chan and Pathak~\cite{ChanP11} improved this guarantee to $\alpha=1.22+\varepsilon$ by employing a detailed analysis to this algorithm. Subsequently, Agarwal and  Sharathkumar in their journal paper~\cite{AS15j} observed that this guarantee is slightly greater than $\frac{\sqrt{2}+1}{2}\approx 1.207$ by presenting an input for $d=3$~\cite{AS15j}; this algorithm outputs the \MEB containing all balls stored.

Moreover, a simple $(1/\sqrt{3})$-approximate two-pass algorithm for the Diameter was given by Egecioglu and Kalantari~\cite{EgeciogluK89}, working in space $\mathcal{O}(d)$.

On the negative side, any randomized algorithm that maintains $\alpha$-$\FP(S)$, $\alpha$-$\MEB(S)$ or $\alpha$-$\coreset(S)$ with probability at least $2/3$ requires $\Omega(\min\{n,\exp(d^{1/3})\})$ space for certain values of $\alpha$. These values are $\alpha<\sqrt{2}(1-2/d^{1/3})$ for $\alpha$-$\FP(S)$ and $\alpha$-$\coreset(S)$ and $a<(1+\sqrt{2})(1/2-1/d^{1/3})$ for $\alpha$-$\MEB(S)$, as shown by Agarwal and Sharathkumar~\cite{AS15j}.

For low $d$, such as $d = O(1)$ or $d = O(\log \log n)$,
the lower bound for $\FP(S)$ does not apply as it is possible to maintain $(1+\varepsilon)$-$\FP(S)$ (or answer $(1+\varepsilon)$-$\FN(x)$ queries) in a poly-logarithmic space, using an optimized version of the sampling approach of~\cite{FrahlingIS08}; this applies even to dynamic streams where points may be deleted.

\section{The Guarded Ball Cover}

The Guarded Ball Cover is a collection $\cB$ of balls that approximately cover all points of $S$. We represent each ball of it by a triplet $B = (c,r,q)$, where $c$ is the center of $B$ (possibly $c\notin S$), $r$ is its radius, and $q$ is a point of $S$. The point $q$ is referred to as the \emph{guard} of $B$.
Our algorithm maintains a coreset $Q$ that consists of the guard points.
We treat the first point $p_1\in S$ specially by always having $p_1\in Q$.

Let $(1+\varepsilon)\cB = \{ (c,(1+\varepsilon)r,q) : (c,r,q) \in \cB\}$ be the collection of the $(1+\varepsilon)$-expansions of the balls in $\cB$. If the arriving point $p\in S$ belongs to $(1+\varepsilon)\cB$, then it is discarded.
Otherwise, a new ball is 
added to $\cB$. 
Finally, all balls of too small radius are removed from $\cB$.

As the space bound is our primary measure, we assume an exact \MEB algorithm, but a good approximation (within, say, a factor of $1+\varepsilon^2/16$) also suffices.

\begin{algorithm}
\caption{Algorithm for processing a new point $p\in S$ (excluding the first point $p_1$)}\label{alg:2prime}\label{alg:2}
\begin{algorithmic}[1]
\If{$p \not\in (1+\varepsilon)\cB$} \Comment{If $p$ is outside of the expansions of all guarded balls}
  \State $(c,r) \leftarrow \MEB(Q\cup\{p\})$ \Comment{Compute new \MEB}
  \State $\cB \leftarrow \{ (c,r,p)\} \cup \{ (c',r',p') \in \cB : r' \ge \varepsilon^{2}r/80\}$
  \Comment{Add new ball, delete small balls}
  \State $Q \leftarrow \{p_1\} \cup \bigcup_{(c'',r'',q)\in \cB} \{q\}$   \Comment{Update coreset}
  
\EndIf
\end{algorithmic}
\label{alg:ours}
\end{algorithm}

The following lemma holds for every \MEB computed in line 2:

\begin{lemma}[Lemma 2.2 in~\cite{BadoiuHI02}]\label{lm:halfspace}
If $B=B(c,r)$ is the \MEB of a set $X$ of points, then any closed half-space containing $c$ also contains a point of $X$ on the boundary of $B$.
\end{lemma}

To analyze the algorithm, we first observe that deleted balls are ``guarded'' by $p_1$.

\begin{lemma}\label{lm:close}
Suppose ball $B$ is deleted when a ball of radius $r$ is added to $\cB$. 
Then, $B$ is contained in a ball of radius $\varepsilon^2 r/40$ centered at $p_1$.
\end{lemma}

\begin{proof}
$B$ contains $p_1$ and when deleted in line 3, it has radius at most $\varepsilon^2 r/80$.
\end{proof}

\noindent
The main technical part is to show that the radii of new balls increase exponentially:

\begin{lemma}\label{lm:expon} 
If $B_{i+1} = (c_{i+1},r_{i+1},p)$ is added to $\cB$ following $B_i = (c_i, r_i, p_i)$, then it holds that $r_{i+1} \ge (1+\varepsilon^{2}/8)\cdot r_i$.
\end{lemma}

\begin{proof}
The proof is similar to that of Lemma 2 in~\cite{AS15j} and Claim 2.4 in~\cite{BadoiuHI02}.
Let $Q$, $\hat{Q}_D$ be the point sets such that
$B_{i+1} = \MEB(Q \cup \{p\})$ and $B_i = \MEB(Q \cup \hat{Q}_D)$. Consider two cases:

%
If $\lVert c_i c_{i+1}\rVert\le 5\varepsilon r_{i}/6$ then $r_{i+1}\ge \lVert c_{i+1}p \rVert\ge \lVert c_ip \rVert -\lVert c_ic_{i+1} \rVert\ge (1+\varepsilon)r_{i}-5\varepsilon r_{i}/6\ge (1+\varepsilon^{2}/6)r_{i}$ (Figure~\ref{fig:1}, left), using the triangle inequality and that $p \not\in (1+\varepsilon)B_i$ (by line 1).

If $\lVert c_i c_{i+1}\rVert > 5\varepsilon r_{i}/6$ then let $h$ be the hyperplane passing through $c_i $ with direction $c_i c_{i+1}$ as its normal and let $h^{+}$ be the halfspace bounded by $h$ that does not contain $c_{i+1}$.  
There is a point $q'\in (Q \cup \hat{Q}_D) \bigcap h^{+}$ at Euclidean distance $r_{i}$ from $c_i$, by Lemma~\ref{lm:halfspace}.
Then,  
$\Vert q'c_{i+1}\rVert\ge (\Vert c_i c_{i+1}\rVert^{2}+\Vert q'c_i \rVert^{2})^{1/2}\ge  ((5\varepsilon r_{i}/6)^{2}+r_{i}^{2})^{1/2} \ge (1+\varepsilon^{2}/4)r_{i}$ (Figure~\ref{fig:1}, right), where the first inequality follows from the cosine law.
By Lemma~\ref{lm:close}, there is a point $q \in Q$ such that $\lVert q q'\rVert \le (\varepsilon^2/40) r_i$.
Hence, $r_{i+1} \ge \lVert q c_{i+1} \rVert \ge \lVert q' c_{i+1} \rVert - \lVert q q' \rVert \ge (1+\varepsilon^2/5) r_i$.

If we use a $(1+\varepsilon^2/16)$-approximate \MEB  (utilizing the algorithm of~\cite{BadoiuClarkson} as a subroutine), then we still have that $r_{i+1} \ge (1+\varepsilon^2/5) r_i/(1+\varepsilon^2/16) \ge (1+\varepsilon^2/8)r_i$.
\end{proof}

\begin{figure}
\begin{center}
\includegraphics[scale=0.8]{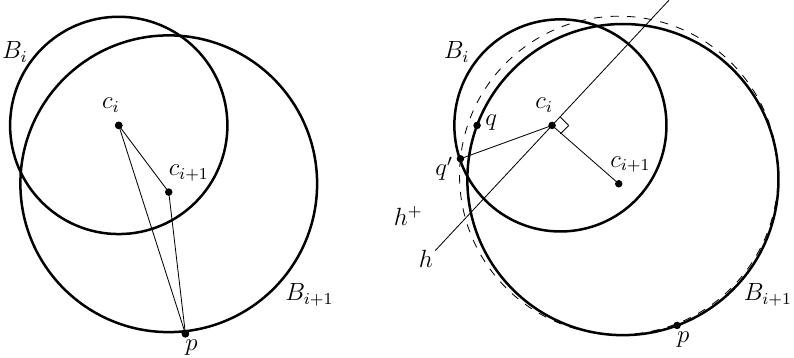}
\caption{Left: Case $\lVert c_i c_{i+1}\rVert \le 5\varepsilon r_{i}/6$, Right: Case $\lVert c_i c_{i+1}\rVert > 5\varepsilon r_{i}/6$ of Lemma~\ref{lm:expon}. The dashed ball is the \MEB of all points, assuming that $q'$ is deleted before $B_{i+1}$ is created.}\label{fig:1}
\end{center}
\end{figure}

Finally, we show that at any time, the $(1+\varepsilon)$-expansion of any deleted ball is contained in the $(1+\varepsilon)$-expansion of each ball created after the deletion of the former ball.
\begin{lemma}
    If ball $\hat{B}$ was deleted then $(1+\varepsilon)\hat{B} \subset (1+\varepsilon)B_i$, for each $B_i \in \cB$ created after the deletion of $\hat{B}$.
    \label{lm:small}
\end{lemma}
\begin{proof}
    Let $B_i \in \cB$.
    By Lemmas~\ref{lm:close} and~\ref{lm:expon}, it is $\hat{B} \subset B(p_1, \varepsilon^{2} r(B_i)/40)$; therefore, we have $(1+\varepsilon)\hat{B} \subset B(p_1, \varepsilon r(B_i)/20)$.
    Since $p_1$ is in $B$, it follows that $(1+\varepsilon)\hat{B} \subset (1+\varepsilon/20)B_i$. 
\end{proof}
Our main result follows from the preceding lemmas: 

\begin{theorem}\label{th:main}
$\cB$ contains $\mathcal{O}((1/\varepsilon^{2})\log(1/\varepsilon))$ balls and $S \subset (1+\varepsilon)\cB$.
\end{theorem}
\begin{proof}
    The first claim follows from Lemma~\ref{lm:expon} and line 3 of the algorithm.
    For the second claim, let $p \in S$. By construction, $p$ is in the $(1+\varepsilon)$-expansion of a ball $B$ that entered $\cB$. By Lemma ~\ref{lm:small}, $(1+\varepsilon)B \subset \cup_{B' \in \cB} (1+\varepsilon)B'$, whether $B$ was deleted or not.
\end{proof}



\section{Applications}

\paragraph*{Farthest Neighbor Queries and Diameter}\label{sec:fn}
We largely follow the analysis of~\cite{AS15j}. 
For a query point $x\in\mathbb{R}^{d}$, the algorithm computes and returns the farthest point in $Q$: $q'=\arg\max_{q\in Q}\lVert qx\Vert$. We show that $\lVert p'x\Vert \le (\sqrt{2}+2\varepsilon)\lVert q'x\Vert$, where $p'=\arg\max_{p\in S}\lVert px\Vert$ is one of the (optimal) farthest points from $x$.  

Let $B_{i}=(c_{i},r_{i})$ be the ball in $\cB$ of greatest radius that contains $p'$ in its $(1+\varepsilon)$-expansion, which exists by Theorem~\ref{th:main}.
Applying the triangle inequality, followed by the inequality $x+y\leq \sqrt{2(x^{2}+y^{2})}$, we have that
\begin{equation}
\label{eq:xp1}    
\lVert xp'\Vert\leq  \lVert xc_{i}\Vert + \lVert c_{i}p'\Vert \leq \lVert x c_{i}\Vert+(1+\varepsilon)r_{i}\leq \sqrt{2}(\lVert x c_{i}\Vert^{2}+r_i^{2})^{1/2}+\varepsilon r_{i}\ .
\end{equation}
By Lemma~\ref{lm:halfspace}, when $B_i$ was created, there was a guard $z$ such that: i) $\lVert z c_i \Vert = r_i$, ii) $\lVert xz \Vert \ge r_i$, and iii) $\angle z c_i p' \ge 90^\circ$.
By iii) and the cosine law, we have 
\begin{equation}
\label{eq:xp2}
\lVert xz\Vert \ge (\lVert x c_{i}\Vert^{2} +r_{i}^{2})^{1/2} \ .
\end{equation}
Combining (\ref{eq:xp1}) and (\ref{eq:xp2}), we get that
\begin{equation}  
\label{eq:xp3}
\lVert xp'\Vert\le \sqrt{2}\lVert xz\Vert+\varepsilon r_{i}\leq (\sqrt{2}+\varepsilon)\lVert xz\Vert \ .
\end{equation}
Note that $\lVert q' x\Vert \ge r_m/2$, where $r_m$ is the radius of the largest ball in $\cB$,
as otherwise there is a ball containing $Q$ of radius less than $r_m$. 
Also, by Lemma~\ref{lm:close}, there is a point $w \in Q$ ($z$ or $p_1$) of Euclidean distance at most $\varepsilon^2 r_m/80 \le \varepsilon^2 \lVert q' x\Vert / 40 $ from $z$. By definition of $q'$, $\lVert wx \Vert \le \lVert q'x\Vert$.
Thus, 
$\lVert xz \Vert \le \lVert wx \Vert + \lVert wz \Vert \le (1+\varepsilon^2/40)\lVert q'x\Vert$.
Combining this with (\ref{eq:xp3}) gives that $q'$ is a $(\sqrt{2}+\varepsilon')$-$\FN(x)$, for $\varepsilon' = 2\varepsilon$.


Finally, for the closely related problem of Diameter, we return $\FN(p)$ for each point $p\in S$. If $\bar{p}=\FN(p)$ and $\lVert p\bar{p}\Vert$ exceeds the stored Diameter, then we replace the old Farthest Pair with the new pair of points $(p,\bar{p})$.

\begin{corollary} For a stream $S$ of points in $\mathbb{R}^{d}$, the Guarded Ball Cover of $\mathcal{O}((1/\varepsilon)^{2}\log(1/\varepsilon))$ stored points answers $(\sqrt{2}+\varepsilon)$-$\FN(x)$ for any query $x\in\mathbb{R}^{d}$, and maintains $(\sqrt{2}+\varepsilon)$-$\FP(S)$. 
\end{corollary}

\paragraph*{Minimum Enclosing Ball}\label{sec:meb}

The following theorem was shown by Chan and Pathak~\cite{ChanP11} and improved the guarantee of the approximate \MEB algorithm of Agarwal and Sharathkumar to $1.22+\varepsilon$ (see~\cite{AS15j}, p. 91):

\begin{theorem}[Theorem 1 in~\cite{ChanP11}, Theorem 1 in~\cite{thesis}]\label{th:chanPathak}
Let $K_{1},...,K_{u}$ be subsets of a point set $S$ in $\mathbb{R}^{d}$, with $B_{i}=\MEB(K_{i})$ such that: i) $r(B_{i})$ is increasing over $i$, and ii) $K_{i}\subset (1+\varepsilon)B_{j}$, for each $i<j$. Then, $r(B)\leq (1.22+\varepsilon)\cdot r(\MEB(S))$, where $B=\MEB(\bigcup_{i=1}^{u} B_{i})$.
\end{theorem}

In our case, $K_{i}$ is the coreset on the boundary of the \MEB computed in line 2 of the algorithm. Therefore, the first requirement of Theorem~\ref{th:chanPathak} holds by Lemma~\ref{lm:expon}. 
The second requirement follows immediately for points in $Q$ and by Lemma~\ref{lm:small} for points deleted from $Q$.

\begin{corollary}
For a stream $S$ of points in $\mathbb{R}^{d}$, the Guarded Ball Cover of $\mathcal{O}((1/\varepsilon^{2})\log(1/\varepsilon))$ stored points maintains $(1.22+\varepsilon)$-$\MEB(S)$. 
\end{corollary}

\paragraph*{Coreset for Minimum Enclosing Ball}\label{sec:cor}

We mostly follow the analysis of~\cite{AS15j}. 
Let $B=(c,r_{m})$ be the most recently created ball added to $\cB$. The $\MEB(S)$ has radius at least $r_m$, since $B$ is an $\MEB$ of a subset of $S$.
We claim that $(\sqrt{2}+\varepsilon)B$ contains all points in $S$, which implies that it forms a $(\sqrt{2}+\varepsilon)$-$\coreset(S)$. 
Namely, we show that each point $y \in S$ has $\lVert y c\Vert \le (\sqrt{2}+\varepsilon')r_m$, for $\varepsilon' = 2\varepsilon$.

Consider a point $y \in S$ that is farthest from $c$.
Let $B_i=(c_i,r_i)$ be a guarded ball that has not been deleted and whose $(1+\varepsilon)$-expansion contains $y$.
By the triangle inequality and the definition of $B_i$,
$\lVert yc\Vert\le \lVert c c_i\Vert + \lVert yc_i\Vert \le \lVert cc_{i}\Vert+(1+\varepsilon)r_{i}$. 
Let $h$ be the hyperplane passing through $c_i$ with direction $c c_i$ as normal and let $h^+$ be the halfspace bounded by $h$ that does not contain $c$.
By Lemma~\ref{lm:halfspace} there is a guard $g$ in $h^+$, and by the cosine law it is
$\lVert gc\Vert\ge \sqrt{\lVert cc_{i}\Vert^{2}+r_{i}^{2}}$.
Then (using the inequality $a+b\le \sqrt{2(a^{2}+b^{2})}$),
\[ \frac{\lVert y c\Vert}{\lVert gc\Vert} \leq \frac{\lVert cc_{i}\Vert+(1+\varepsilon)r_{i}}{\sqrt{\lVert cc_{i}\Vert^{2}+r_{i}^{2}}} \le \sqrt{2}+\varepsilon \ . \] 
By Lemma~\ref{lm:close}, there is a guard $q \in Q$ with $\lVert q g\Vert \le \varepsilon^2 r_m/40$, so by the triangle inequality,
$\lVert gc\Vert \le \lVert qc\Vert + \lVert gq\Vert \le (1 + \varepsilon^2/40)r_m$.
Hence, $\lVert y c\Vert \le (\sqrt{2}+\varepsilon)(1+\varepsilon^2/40) r_m \le (\sqrt{2}+2\varepsilon) r_m$.

\begin{corollary}\label{cor:new4}
For a stream $S$ of points in $\mathbb{R}^{d}$, the Guarded Ball Cover of $\mathcal{O}((1/\varepsilon)^{2}\log(1/\varepsilon))$ stored points maintains $(\sqrt{2}+\varepsilon)$-$\coreset(S)$. 
\end{corollary}

\section{Lower Bound for Farthest Pair and Farthest Neighbor Queries}

We show that computing a $\sqrt{2}$-approximation (with $\varepsilon =0$) of Farthest Neighbor queries or Farthest Pair is impossible in the streaming model without returning points outside of $S$.

This applies to the computational model of ``coreset-based algorithms'' in which the space bound is counted in the number of input points stored and the algorithm must return an input point upon a query (or two input points for the Farthest Pair);
crucially, once a point is deleted from memory, it cannot be retrieved.
This model is akin to comparison-based model in sorting or selection,
as used in~\cite{CormodeV20} for streaming lower bounds for quantile estimation.

\begin{theorem}
	For any $\varepsilon > 0$ and $d = \Omega(1/\varepsilon)$,
	any coreset-based randomized streaming algorithm answering approximate Farthest Neighbor queries or maintaining the approximate Farthest Pair in $\mathbb{R}^{d}$ with multiplicative error $\le \sqrt{2}+\varepsilon$, 
	has to store $\Omega(1/\varepsilon)$ points.
\end{theorem}

\begin{proof}
	We use the easy direction of Yao's minimax principle and design a distribution over instances (points and Farthest Neighbor queries) so that 
	any deterministic streaming algorithm using space $o(1/\varepsilon)$ will, with high constant probability, answer a $\FN(q)$ query incorrectly,
	i.e., the point returned on the query $q$ will be more than $(\sqrt{2}+\varepsilon)$-factor closer to $q$ than the farthest point. The same argument applies to Farthest Pair.
	
	Supposing, without loss of generality,\ that $2/\varepsilon\in \mathbb{Z}$, we insert $k := 1/(2\varepsilon) + 1$ points of the standard basis, i.e., $\mathbf{e}_i = (0, 0, \dots, 1, 0, \dots, 0)$, where
	the $i$-th coordinate is $1$, for $i = 0, \dots, k-1$ (the order of insertions does not matter).
	The random part of the construction is to choose $j\in \{0, \dots, k-1\}$ uniformly at random 
	and make a Farthest Neighbor query for point $q_j = (2\varepsilon, 2\varepsilon, \dots, -1, 2\varepsilon, \dots, 2\varepsilon, 0, \dots, 0)$,
	where the coordinate $j$ is $-1$ and only the first $k$ coordinates are not $0$.
	Clearly, the farthest point from $q_j$ is $e_j$ and their Euclidean distance is $\sqrt{4 + (k-1)\cdot 4\varepsilon^2} = \sqrt{4 + 2\varepsilon}$,
	using the choice of $k$.
	
	However, with some constant probability, point $e_j$ is not stored as the algorithm stores $o(1/\varepsilon) = o(d)$ points.
	Conditioning on the event that $e_j$ is not stored, the algorithm needs to answer the query with a point $e_i$ for $i\neq j$.
	However, the Euclidean distance between $q_j$ and $e_i$ with $i\neq j$ is 
	$\sqrt{(1-2\varepsilon)^2 + 1 + (k-2)\cdot 4\varepsilon^2} = \sqrt{2 - 4\varepsilon + (k-1)\cdot 4\varepsilon^2} = \sqrt{2 - 2\varepsilon}$,
	using the choice of $k$.
	It follows that the approximation ratio of the algorithm is at least 
	$\sqrt{(4 + 2\varepsilon)/(2 - 2\varepsilon)}>\sqrt{2} + \varepsilon$.
\end{proof}

\bibliography{references}

\appendix

\section{A Note on Lemma~2 in~\cite{AS15j}}
\label{app:AS15issue}
We report on a possible issue in the proof of Lemma~2 in~\cite{AS15j} (a similar issue appears in the conference version~\cite{AS15}). Lemma~2 in~\cite{AS15j} states that for any $i$, $r(B_{i+1})\ge (1+\varepsilon^2 / 8)\cdot r(B_i)$, where $B_i$ and $B_{i+1}$ are two consecutive balls of the Blurred Ball Cover. The property that the radii of balls increase geometrically is crucially needed to bound the space requirements.

The proof of Lemma~2 goes as follows: 
Let ball $B_{i+1}=\APPROXMEB(\bigcup_{j\le i} K_{j}\cup A),\varepsilon/3)$, where $K_{j}$ is the coreset of $B_{j}$, $A$ is a buffer of incoming points, and \APPROXMEB is the subroutine of~\cite{BadoiuClarkson} that computes it;
the computation of $B_{i+1}$ is triggered because each point in $A$ was not contained in any $(1+\varepsilon)$-expansion of a ball $B_j$ for $j\le i$. Define ball $B'$ as the \MEB of $\bigcup_{j\le i} K_{j}\cup A$. 
It is subsequently claimed that $r(B_{i+1})\ge r(B')$, without a proof. 
However, $B'$ is the \MEB of $\bigcup_{j\le i} K_{j}\cup A$, while $B_{i+1}$ is an \emph{approximate} \MEB for $\bigcup_{j\le i} K_{j}\cup A$, namely the $(1+\varepsilon/3)$-expansion of $B_{i+1}$ contains all these points but a smaller expansion of $B_{i+1}$ may not contain them all. 

We observe that this can be fixed by the computation of a tighter \MEB approximation
$B_{i+1}=\APPROXMEB(\bigcup_{j\le i} K_{j}\cup A),\varepsilon^2/16)$. That, however, may result in a coreset $K_{i+1}$ of size $\Theta(\varepsilon^{-2})$, which (unlike our approach) increases the space bound in~\cite{AS15j} to $\mathcal{O}(\varepsilon^{-4} \log(1/\varepsilon))$.
We are unaware of a fix that does not affect the space bound or guarantees.

\end{document}